\documentclass[11pt]{amsart}

\usepackage[centering]{geometry}
\usepackage[hidelinks]{hyperref}

\numberwithin{equation}{section}

\theoremstyle{definition}
\newtheorem*{dfn}{Definition}
\newtheorem{pro}{Proposition}[section]
\newtheorem{lem}[pro]{Lemma}
\newtheorem*{rmk}{Remark}
\newtheorem*{emp}{Example}

\newcommand{\eg}{\textit{e.g.}}

\newcommand{\cf}{\textit{cf.}}
\newcommand{\etc}{\textit{etc.}}
\newcommand{\wrt}{\text{w.r.t.}}
\newcommand{\etal}{\textit{et\ al.}}
\newcommand*{\markdef}[1]{{\normalfont\textbf{#1}}}
\newcommand*{\argmax}{\operatornamewithlimits{argmax}}

\newcommand*{\R}{\mathbb{R}}
\newcommand*{\Z}{\mathbb{Z}}

\newcommand*{\E}{\mathbb{E}}
\newcommand*{\D}{\mathbf{D}}
\newcommand*{\dd}{\mathbf{d}}
\newcommand*{\ww}{\mathbf{w}}
\newcommand*{\vv}{\mathbf{v}}
\newcommand*{\BL}{\text{BL}}
\newcommand*{\BM}{\text{BM}}
\newcommand*{\SL}{\text{SL}}
\newcommand*{\SM}{\text{SM}}
\DeclareMathOperator{\id}{id}

\begin{document}

\title[Dynamical Trading Mechanism]{Dynamical Trading Mechanism in Limit Order Markets}
\author[S. Wang]{Shilei Wang}
\curraddr{Ca' Foscari University of Venice, Venice, Italy}
\email{shilei.wang@unive.it}
\date{March 2013}

\begin{abstract}
This work's purpose is to understand the dynamics of limit order books in order-driven markets. We try to illustrate a dynamical trading mechanism attached to the microstructure of limit order markets. We capture the iterative nature of trading processes, which is critical in the dynamics of bid-ask pairs and the switching laws between different traders' types and their orders. In general, after introducing the atomic trading scheme, we study a general iterated trading process in both combinatorial and stochastic ways, and state a few results on the stability of a dynamical trading system. We also study the controlled dynamics of the spread and the mid-price in an iterated trading system, when their movements, generated from the dynamics of bid-ask pairs, are assumed to be restricted within some extremely small ranges.
\end{abstract}

\maketitle

\section{Introduction}

A market as a set of institutional arrangements is efficient in information transmission, wealth allocation, and value revelation. Whatever the intentions of agents involved in a market, their individual behaviors, either strategic or non-strategic, collectively shape various trading processes in the market. Hence, the microstructure of a market may have its representation as a corresponding trading mechanism, which exactly defines the trading games in the market. We can then catch proper insights into the market microstructure by studying on the level of trading processes. 

As for normal goods, agents can immediately exchange them through typical markets, in which the money as a medium is used to measure their intrinsic values. If the demand and supply are not proportionally concentrated or distributed in a market, auction or bilateral bargaining as trading mechanisms can be used to coordinate those unmatched parts in the market. But in a high-frequency electronic market with financial assets, the situation is more complex. Note that speculative activities on the market will make its log-scaled price move roughly in Brownian motion (see \eg\ Osborne \cite{osborne59}), and thus its volatility is propositional to the square root of the time. As a result, the liquidity in a financial market is valuable, and it makes any good at a specific time multi-dimensional. So the value of any specific good depends not only on its price, but on how long it has stayed on the market. Such a multi-dimensional good has a monopoly power, in the sense that its provision is decentralized on the market at each time, so its demand and supply are eventually concentrated. A more complicated trading mechanism, namely continuous double auction, can be used to smooth trading processes in those financial markets. The financial markets applying continuous double auction mechanism are traditionally called order-driven markets, or limit order markets.

The very first study on limit order markets is a byproduct from a paper on the cost of transacting carried by Demsetz \cite{demsetz68}, in which he defines the bid-ask spread as a markup ``paid for predictable immediacy of exchange in organized markets'' (see \cite{demsetz68}, p.\,36), and considers it as an important source in transaction costs. However, more later studies on limit order markets focus on their microstructure, rather than solely on the role of asymmetric information in trading decisions. A possible reason for this trend is that the dynamics of trading decisions and the evolution of markets are considered to be more important than the stable states in equilibrium, especially after the crash of October 1987, when all major world markets declined substantially, from the minimum $11.4\%$ in Austria to the maximum $45.8\%$ in Hong Kong (\cf\ Sornette \cite{sornette03}, p.\,5).

Nowadays, most stock exchanges, more or less as order-driven markets, adopt electronic order-driven platforms, partially because of the development of electronic communications networks and adaptive regulations on order handling rules. For instance, the Euronext Paris (formerly known as the Paris Bourse, and recently the NYSE Euronext), the Tokyo Stock Exchange, the Toronto Stock Exchange, the NYSE, and the NASDAQ are typical limit order markets, in which orders are essentially executed through electronic systems. They provide opportunities to generate stylized empirical facts on limit order markets. For example, Lehmann and Modest \cite{lehmann94} study the trading mechanism and the liquidity in the Tokyo Stock Exchange, Biais, Hillion, and Spatt \cite{biais95} study the limit order book and the order flow in the Paris Bourse, Harris and Hasbrouck \cite{harris96} measure the performance of SuperDOT traders in the NYSE, and Al-Suhaibani and Kryzanowski \cite{al-suhaibani00} analyze the order book and the order flow in the Saudi Stock Exchange (Tadawul).

The theoretical understandings on limit order markets can mainly be divided into two approaches. One is to explain the performance of a market and its stable states by modeling traders' behaviors in equilibrium. The other is to consider the market as a ``super-trader'' with zero intelligence, and explain or predict her behaviors statistically. The first approach understands traders' strategic behaviors and generates testable implications for a market by capturing traders' different attributes, for instance, being informed versus being uninformed used by Kyle \cite{kyle85}, and Glosten and Milgrom \cite{glosten85}, time preference used by Parlour \cite{parlour98}, and patience versus impatience used by Foucault, Kadan, and Kandel \cite{foucault05}, \etc\ The second approach analyzes the market by assuming that there exist a few statistical laws in the dynamics of the market. It is somehow able to catch certain profiles of a market, say, notably fat tails of the price distribution, concavity of the price impact function, scaling law of the spread \wrt\ the orders, and so on (see \eg\ Smith \etal\ \cite{smith03}, and Farmer, Patelli, and Zovko \cite{farmer05}, \etc).

Our work will take a middle position between these two approaches. We agree that the performance of a market and its evolution are determined by the behaviors of rational traders involved in the market, but in the meantime, we also keep in mind that statistical mechanics could be important in the price dynamics. Therefore, we decide to consider a large population of traders in a limit order market, and assume that each trader in the sample is rational, such that her behaviors are strategically optimal in a certain game-theoretic framework. We assume that each trader's rational trading decision can be represented in terms of the distribution of the order book on the market. So the individual rationality and the collective rationality are supposed to be determined interchangeably in the limit order market. We want to study how individual traders affect the price dynamics in the market, how a certain sequential trading process influences the stability of the market, how the introduced randomness in the trading process enhances the systemic stability, and why the market can evolve more predictably if we control some factors in the market. In general, we try to clarify a dynamical trading mechanism in the limit order market.

\section{Atomic Trading Scheme}\label{sec2}

\subsection{Preliminary Framework}\label{sec2.1}

We consider a generic order-driven market with a large population of traders, say $N$, in which the attributes of traders can be represented by their trading directions and demands of the liquidity. As usual, the trading direction is either selling or buying initiation, while the liquidity demand determines the type of any submitted order, which is either a limit order or a market order. So $N$ can be partitioned into four different groups by these two kinds of binary classification. Any trader drawn from $N$ will be in one and only one of the following four groups: (i) buyers submitting limit orders, (ii) buyers submitting market orders, (iii) sellers submitting limit orders, and (iv) sellers submitting market orders. To simplify the words used for descriptions, we denote the types of traders in the four groups by ``BL'', ``BM'', ``SL'', and ``SM'' respectively.

A trader's relative valuation on the equities in the market determines her trading role. If she has a value higher than the average level of her peer traders in the same market, she will be more likely to be a buyer. If her value is lower than the average, she will be more likely to be a seller. The trader's trading strategy is determined by her demand of the liquidity. If the bid-ask spread is sufficiently large, which means that the market lacks liquidity, then she will provide it on the market, and thus her order will be a limit one. On the other hand, if the bid-ask spread is very small, which means that the market is full of liquidity, then she will consume it, and thus her order will be a market one. In the case of limit order, we assume that it will definitely improve the attractiveness of the best quotes for a corresponding part of traders in the market. In the case of market order, by the normal principle of price-time priority, we assume that it will hit exactly the best bid or best ask, since the best quotes have the priority to meet any new market order.

In a real limit order market, however, the trading process is naturally complicated, simply as there exists a variety of possible trading strategies. Traders can submit orders, they can cancel submitted orders, they can submit hidden orders, they can split orders, and so on. So the limit order book can not be fully characterized only by its bid and ask prices, but the trading volumes and the depth of the order book are also very important. For instance, many limit orders may queue behind the quotes of the order book, and some market orders may have trading volumes not exactly equal to the depth at the best quotes. In both cases, the bid and ask prices will not be definitely changed after a new order, as the new limit order may just improve the depth at the best quotes of the order book, or it may not fully consume the shares at the quotes. 

If we focus on the dynamics of the quotes on the limit order market, these additional considerations will just make the time needed for an updated quote, as a jump, become uncertain. We assume that the order book's depth is equal to the trading volume of any new order. So we can consider the normal time domain $t+\Z$, rather than $t+\E$, where $t$ is the initial time, and $\E$ is an arbitrary subset of $\Z$ when there are jumps. As a result, we have only two critical trading executions for the price dynamics in the limit order market, namely, limit orders and market orders with their trading volumes equal to the order book's depth. Consequently, all the four types of traders in $N$ can definitely affect the market at each trading period.

\begin{dfn}
A trader is called a \markdef{marginal trader}, if the quotes of the limit order book are updated after her submitted order.
\end{dfn}

Since we have introduced the marginal trader in the limit order market, we now can look at the order book in a formal way. Suppose the limit order book has a best bid $b\in\R_+$ and a best ask $a\in\R_+$ at any time. We denote them by the bid-ask pair $(b,a)$. The bid-ask spread $s$, and the mid-price or the quote midpoint $m$ are determined by $(b,a)$, say, $s=a-b$ and $m=(b+a)/2$. Since there exists a tick size, denoted by $\tau>0$, as the minimal change of the prices in the market, we have $b\neq a$ and $s\geq\tau$. Moreover, we set a stricter condition for the lower bound of the spread in the market. We claim that there exists a lower bound $\underline{s}>\tau$ for the bid-ask spread at any time in the limit order market, or $s\geq\underline{s}$.

Naturally, there also exists an upper bound for the best ask $a$, say $\overline{a}$, since the value of any security on the market is limited for all trader in $N$. Besides that, note that $b\geq 0$, otherwise, there would be no demand in the market, as the inverse of the best bid would exist as a part of the ask side of the market. As a result, the pair $(b,a)$ should be located within a compact domain $W\subset\R_+^2$, which is defined by $b\geq 0$, $a\leq\overline{a}$, and $a-b\geq\underline{s}$. In the $b$-$a$ plane, $W$ can be geometrically represented as a triangle, whose vertices are $(0,\underline{s})$, $(0,\overline{a})$, and $(\overline{a}-\underline{s},\overline{a})$.

Let the time domain be $\Z$. At each time $t\in\Z$, the best bid and best ask are denoted by the bid-ask pair $(b_t,a_t)=\ww_t$, or the bid-ask vector $(b_t,a_t)'=\dd_t$. If the time $t$ is in the past, we use $\dd_t=(b_t,a_t)'$ as usual. If the time $t$ is in the future, we use a random variable $\D_t$, such that $\D_t=(B_t,A_t)'$, where $B_t$ and $A_t$ are stochastic forms of $b_t$ and $a_t$ respectively. For any given bid-ask pair $(b_t,a_t)$, we define the bid-ask spread by $s_t=a_t-b_t$, and the mid-price by $m_t=(b_t+a_t)/2$. If the time index is not so important in some cases, we also use $\ww$ and $\dd$ to represent respectively the pair $(b,a)$ and the vector $(b,a)'$, and the derived spread $s=a-b$, and mid-price $m=(b+a)/2$. 

Suppose we are being at the time $t\in\Z_+$ such that the events at $t$ have just happened, so we have a sequence of historical information,
\[
\{\dd_s,\ s\in\Z\ \text{and}\ s\leq t\}=\{\dd_{-\infty},\dotsc,\dd_{t-1},\dd_{t}\}, 
\]
and a sequence of uncertain or random information, which forms a stochastic process,
\[
\{\D_{s},\ s\in\Z\ \text{and}\ s>t\}=\{\D_{t+1},\D_{t+2},\dotsc,\D_\infty\}.
\]
Generally speaking, we want to study how the stochastic variable $\D_t$ can be determined by traders' strategies and previous information on the market in the form of $\dd_t$, or equivalently $s_t$ and $m_t$.

As for the whole order book, we suppose that at each time $t$ the adjacent quotes on both sides of the order book are equally distributed over $[0,b_t]$ and $[a_t,\overline{a}]$ respectively. We assume the difference between adjacent quotes on the same side of the order book is proportional to the bid-ask spread of the order book, as empirically found by Biais, Hillion, and Spatt \cite{biais95} using data from the Paris Bourse, and Al-Suhaibani and Kryzanowski \cite{al-suhaibani00} with facts from the Saudi Stock Exchange. Furthermore, the difference of adjacent quotes is equal to $\alpha(a_t-b_t)$ on the ask side of the book, and $\beta(a_t-b_t)$ on the bid side of the book, where $\alpha,\beta\in(0,1)$, which implies that the order book is slightly concave. Then the upper quote next to the best ask $a_t$ is equal to $a_t^+=a_t+\alpha(a_t-b_t)$, and the lower quote next to the best bid is equal to $b_t^-=b_t-\beta(a_t-b_t)$. We assume that the ratio of $\beta:1:\alpha$ in the differences between the four consecutive prices $b_t^-,b_t,a_t,a_t^+$ in the order book is an indicator for traders' collective rationality in the limit order market.

Note that $\alpha$ and $\beta$ can not be very close to $0$, otherwise the limit order appears there is not a rational decision, as the obtained price-time priority and the cost of the liquidity are not balanced optimally there. Also $\alpha$ and $\beta$ will not be close to $1$, otherwise more rational limit orders would emerge between adjacent quotes to exploit the capacity of the book there. In fact, $\alpha$ and $\beta$ are roughly $0.5$, as suggested statistically by Biais, Hillion, and Spatt \cite{biais95}. 

\subsection{Evolution of Bid-Ask Pairs}\label{sec2.2}

Suppose that a marginal trader gets to the market at the time $t$ just after the bid-ask pair $(b_t,a_t)$ forms, thus she has all the information at and before the time $t$ in the market. Her trading decision at $t$ can be denoted by the price $p_t$. If she submits a limit order, $b_t<p_t<a_t$, and if she submits a market order, $p_t\geq a_t$ or $p_t\leq b_t$. If she is a buyer, $p_t\leq a_t$, and if she is a seller, $p_t\geq b_t$. Recall that we have assumed the order book's depth at the quotes is equal to the trading volume of any new order, so a market order will definitely clear one of the limit orders at the best quotes, and the cleared order will be replaced by a less attractive limit order. A limit order will decrease the bid-ask spread, and surely improve one of the best quotes. Consequently, any type of marginal trader must change the states in the limit order book.

\subsubsection*{\bf Type BL}

If the marginal trader is a buyer and submits a limit order, then the best bid $b_{t+1}$ at the time $t+1$ will be $p_t=b_t+\rho$, where $\rho>0$, and the best ask $a_{t+1}$ at the time $t+1$ will remain unchanged. Suppose the marginal trader's decision process can be described by her maximizing the utility function with respect to $p_t$ or equivalently $\rho$, say $u(p_t)=u(\rho+b_t)$, where $u$ is concave in $\rho$ and of class $C^1$. By the requirement of rationality, we have $\rho^\ast\in\argmax_\rho u(p_t)$, such that $\rho^\ast>0$ and $a_t-b_t-\rho^\ast\geq\underline{s}$. But according to our assumption on the rationality of the limit order book, we directly have the following equation with the optimal $\rho^\ast$,
\begin{equation}
\rho^\ast=\beta(a_t-b_t-\rho^\ast),
\end{equation}
so $\rho^\ast=\frac{\beta}{1+\beta}(a_t-b_t)$. Thus $b_{t+1}=\frac{1}{1+\beta}b_t+\frac{\beta}{1+\beta}a_t$, and $a_{t+1}=a_t$, or we have $\dd_{t+1}=S_1\dd_t$, where
\[
S_1=
\begin{pmatrix}
\frac{1}{1+\beta} & \frac{\beta}{1+\beta}\\ 
0 & 1
\end{pmatrix},
\]
and $\dd_{t+1}$ is the realization of the random variable $\D_{t+1}$. After this trade of limit order on the bid side, we obtain a new bid-ask spread,
\begin{equation}
a_{t+1}-b_{t+1}=\frac{1}{1+\beta}(a_t-b_t).
\end{equation}

Note that $a_t-b_t-\rho^\ast\geq\underline{s}$ and $\rho^\ast=\frac{\beta}{1+\beta}(a_t-b_t)$ will generate $a_t-b_t\geq(1+\beta)\underline{s}$. So the marginal trader is of type BL only if the original bid-ask spread is greater than or equal to $(1+\beta)\underline{s}$, and if she is of type BL, she will choose an optimal improvement $\rho^\ast\geq\beta\underline{s}$. 

\subsubsection*{\bf Type BM}

If the marginal trader is again a buyer, but now she submits a market order hits $a_t$, then the best ask $a_{t+1}$ at the next period will be $a_t+\alpha(a_t-b_t)$, and the best bid $b_{t+1}$ will remain the same. Here, her decision process is neglected by the principle of price-time priority employed in the limit order market. So her decision set is a singleton, say $\{a_t\}$, which implies that $p_t=a_t$. Consequently, we have $b_{t+1}=b_t$ and $a_{t+1}=-\alpha b_t+(1+\alpha)a_t$, or $\dd_{t+1}=S_2\dd_t$, where
\[
S_2=
\begin{pmatrix}
1 & 0\\ 
-\alpha & 1+\alpha
\end{pmatrix}.
\]
After this trade of market order on the bid side, the new bid-ask spread is
\begin{equation}
a_{t+1}-b_{t+1}=(1+\alpha)(a_t-b_t).
\end{equation}

\subsubsection*{\bf Type SL}

If the marginal trader is a seller, and she submits a limit order, then the best ask $a_{t+1}$ at the time $t+1$ will be $a_t-\theta$, where $\theta>0$, and the best bid $b_{t+1}$ at the time $t+1$ will remain the same as $b_t$. Similar with the type BL, this type of marginal trader's decision process can be described as maximizing her utility $v(p_t)=v(-\theta+a_t)$, where $v$ is convex in $\theta$ and of class $C^1$. Her decision will admit the optimal choice $\theta^\ast\in\argmax_\theta v(p_t)$, such that $\theta^\ast>0$ and $a_t-\theta^\ast-b_t\geq\underline{s}$. By the represented rationality on the ratio of differences between adjacent quotes around the best quotes in the order book, we have
\begin{equation}
\theta^\ast=\alpha(a_t-\theta^\ast-b_t),
\end{equation}
so $\theta^\ast=\frac{\alpha}{1+\alpha}(a_t-b_t)$. Thus $b_{t+1}=b_t$, and $a_{t+1}=\frac{\alpha}{1+\alpha}b_t+\frac{1}{1+\alpha}a_t$, or more concisely $\dd_{t+1}=S_3\dd_t$, where
\[
S_3=
\begin{pmatrix}
1 & 0\\ 
\frac{\alpha}{1+\alpha} & \frac{1}{1+\alpha}
\end{pmatrix}.
\]
After this trade of limit order on the ask side, the bid-ask spread will be updated to
\begin{equation}
a_{t+1}-b_{t+1}=\frac{1}{1+\alpha}(a_t-b_t).
\end{equation}

Similar with the type BL, $a_t-\theta^\ast-b_t\geq\underline{s}$ and $\theta^\ast=\frac{\alpha}{1+\alpha}(a_t-b_t)$ will generate $a_t-b_t\geq(1+\alpha)\underline{s}$. So the marginal trader is of type SL only if the bid-ask spread is greater than or equal to $(1+\alpha)\underline{s}$, and if she is of type SL, she will choose an optimal improvement $\theta^\ast\geq\alpha\underline{s}$.

\subsubsection*{\bf Type SM}

If the marginal trader is a seller and submits a market order hits $b_t$, then the best bid $b_{t+1}$ at the time $t+1$ will be $b_t-\beta(a_t-b_t)$, and the best ask $a_{t+1}$ at the next period will remain unchanged. Her decision is restricted to choosing $p_t$ to maximize her utility subject to $p_t\in\{b_t\}$, so the optimal choice is $p_t=b_t$. We have $b_{t+1}=(1+\beta)b_t-\beta a_t$ and $a_{t+1}=a_t$, or $\dd_{t+1}=S_4\dd_t$, where
\[
S_4=
\begin{pmatrix}
1+\beta & -\beta\\ 
0 & 1
\end{pmatrix}.
\]
After this trade of market order on the ask side, we have
\begin{equation}
a_{t+1}-b_{t+1}=(1+\beta)(a_t-b_t).
\end{equation}

\subsection{Switching Laws}\label{sec2.3}

From the analysis on the trading behaviors of different marginal traders, we can see how the bid-ask spread develops as a Markov process. Namely, a limit order will change the bid-ask spread from $s$ to $s/(1+\alpha)$ or $s/(1+\beta)$, while a market order will change it from $s$ to $(1+\alpha)s$ or $(1+\beta)s$. So any initial bid-ask spread $s$ will converge to the infinity after sufficiently many market orders, and any $s$ will converge to $0$ after infinitely many limit orders. But $s$ has a lower bound $\underline{s}$ and an upper bound $\overline{a}$ for all bid-ask pair in $W$. When $s$ is too close to its bounds, the market will make $s$ move away from them. Intuitively, we say there exist some kind of ``gravitational forces'' controlling the appearance of different types of orders in the market. Moreover, such forces may also generate corresponding switching rules between different types of traders. 

Noting that the force is related with the bid-ask spread, rather than directly with the bid-ask pair, we will focus on the bid-ask spread as an indicator for it. If the bid-ask spread is sufficiently large, the hidden force will attract limit-type traders, say BL and SL, and repel market-type traders, say BM and SM, (but don't necessarily reject them until the spread is very close to the upper bound $\overline{a}$). On the other hand, if the bid-ask spread is extremely small, then such a hidden force will attract market-type traders and repel limit-type traders, and it will reject limit-type traders once the spread is close enough to $\underline{s}$.

The bid-ask spread at the time $t$ is $s_t=a_t-b_t$, where $s_t\in[\underline{s},\overline{a}]$ for all $(b_t,a_t)\in W$. The market only accepts traders of type BM and type SM, when
\[
s_t<\min\{(1+\alpha)\underline{s},(1+\beta)\underline{s}\}=(1+\min\{\alpha,\beta\})\underline{s},
\]
otherwise, $s_{t+1}<\underline{s}$ if it accepts limit-type traders. The market attracts traders of type BL and type SL, but may also accept traders of type BM and type SM, when
\[
s_t\geq(1+\max\{\alpha,\beta\})\underline{s}.
\]
In the remaining interval of $s_t$, say,
\[
(1+\min\{\alpha,\beta\})\underline{s}\leq s_t<(1+\max\{\alpha,\beta\})\underline{s}, 
\]
which types of traders will be attracted by the market depend on the ratio of $\alpha$ and $\beta$. If $\alpha>\beta$, or $\alpha/\beta\in(1,\infty)$, the market will only accept the trader of type BL. If $\beta>\alpha$, or $\alpha/\beta\in(0,1)$, the market will only accept the trader of type SL. If $\alpha=\beta$, or $\alpha/\beta=1$, then $\min\{\alpha,\beta\}=\max\{\alpha,\beta\}$, and hence such an interval does not exist. 

In sum, the possible types of marginal traders in a market with a bid-ask spread $s_t\in[\underline{s},(1+\min\{\alpha,\beta\})\underline{s})$ are BM and SM, while BL and SL would be rejected by the market. A market with a spread $s_t\in[(1+\min\{\alpha,\beta\})\underline{s},\underline{a}]$ would prefer traders of type BL and/or type SL to the type BM and type SM. 

We can disregard the interval between $(1+\alpha)\underline{s}$ and $(1+\beta)\underline{s}$, if
\[
\frac{\vert\alpha-\beta\vert\underline{s}}{\overline{a}-\underline{s}}=\frac{\vert\alpha-\beta\vert}{\overline{a}/\underline{s}-1}\approx 0.
\]
Since $\overline{a}\gg\underline{s}$ in most normal limit order markets, and $\alpha,\beta\in(0,1)$, that condition will be always satisfied. As a result, we assume $\alpha=\beta$, which exactly makes such an intermediary interval disappear. We can then consider a simplified but again general enough limit order book, in which we have
\[
[\underline{s},\overline{a}]=[\underline{s},(1+\alpha)\underline{s})\cup[(1+\alpha)\underline{s},\overline{a}].
\]
In that limit order market, the spread $s_t$ will be updated to $s_{t+1}=s_t/(1+\alpha)$ by a marginal trader of type BL or SL, or updated to $s_{t+1}=s_t(1+\alpha)$ by a marginal trader of type BM or SM. If the spread moves into $[\underline{s},(1+\alpha)\underline{s})$, it will move towards $[(1+\alpha)\underline{s},\overline{a}]$ in the following periods. If the spread is very close to $\overline{a}$, it will bounce away immediately. In general, we establish the following result on the capacity of limit-type traders at different regions of the order book.

\begin{pro}\label{pro2}
The maximal number of limit-type traders, who can be continuously accepted by a limit order market with a spread $s$, is determined by the floor function
\[
z(s)=\left\lfloor\frac{\log s-\log\underline{s}}{\log(1+\alpha)}\right\rfloor.
\]
\end{pro}

\begin{proof}
Define a sequence of consecutive intervals, namely, $[(1+\alpha)^i\underline{s},(1+\alpha)^{i+1}\underline{s})$, where $i\in\{0,1,\dotsc,n\}$, and 
\[
n=\max\{i\in\Z:(1+\alpha)^i\leq\overline{a}\}-1.
\]
For all $s\in[\underline{s},(1+\alpha)^{n+1}\underline{s})$, there exists a unique $j(s)\in\{0,1,\dotsc,n\}$ such that
\[
s\in[(1+\alpha)^{j(s)}\underline{s},(1+\alpha)^{j(s)+1}\underline{s}).
\]

We want to show that $z(s)=j(s)$ by induction. If $j(s)=0$, then $s\in[\underline{s},(1+\alpha)\underline{s})$, and the limit order market will reject the limit-type order, so $z(s)=0$. Assume $z(s)=j(s)$ is true for all $j(s)\leq k$, and consider $j(s)=k+1$ such that $s\in[(1+\alpha)^{k+1}\underline{s},(1+\alpha)^{k+2}\underline{s})$. After a marginal trader of type BL or type SL comes to the market, $s$ will be updated to $s'=s/(1+\alpha)\in[(1+\alpha)^k\underline{s},(1+\alpha)^{k+1}\underline{s})$. By the assumption, we know $z(s')=k$, so $z(s)=z(s')+1=k+1=j(s)$. 

If $s\in[(1+\alpha)^{n+1}\underline{s},\overline{a}]$, we have $s\in[(1+\alpha)^{n+1}\underline{s},(1+\alpha)^{n+2}\underline{s})$, as $(1+\alpha)^{n+2}\underline{s}>\overline{a}$. Here $j(s)=n+1$, so we have $z(s)=n+1=j(s)$ by induction, simply as $z(s')=j(s')$ for all $j(s')=n$.

Therefore, $z(s)$ satisfies
\[
(1+\alpha)^{z(s)}\underline{s}\leq s<(1+\alpha)^{z(s)+1}\underline{s},
\]
and thus we have
\[
z(s)\leq\frac{\log s-\log\underline{s}}{\log(1+\alpha)},\ \text{and}\ z(s)+1>\frac{\log s-\log\underline{s}}{\log(1+\alpha)},
\]
which exactly define the floor function.
\end{proof}

We can alternatively state that there is an exponential law in the relationship between the bid-ask spread $s$ and the market's capacity of accepting limit-type traders $n$,
\begin{equation}
s=(1+\alpha)^n\underline{s},\ \text{where}\ n\in[z(s),z(s)+1).
\end{equation}

An intuitive but nature remark from this result is that the probability distribution function $f(s)$ of the appearance of limit-type traders is positively linear in $\log s$, namely, 
\[
f(s)=k_1\log s+k_2,\ \text{if}\ s\in[(1+\alpha)\underline{s},\overline{a}],
\]
where $k_1>0$ and $k_2$ are constants depending on the values of $\underline{s},\overline{a},\alpha$, and
\[
f(s)=0,\ \text{if}\ s\notin[(1+\alpha)\underline{s},\overline{a}].
\]

Suppose there exists an interval $((1-\gamma)\overline{a},\overline{a}]$, where $0<\gamma<1$ and $(1-\gamma)(1+\alpha)\leq 1$\footnote{This inequality gives a necessary condition for a zero-capacity of the limit order market in accepting market-type orders. If we assume $(1-\gamma)(1+\alpha)=1$, there would be less room for interesting analysis.}, or equivalently $\alpha/(1+\alpha)\leq\gamma<1$, under which the market will never accept any more market-type trader. We can then state a result on the capacity of market-type traders in the market, whose unwritten proof is similar to that of Proposition \ref{pro2}. 

\begin{pro}
The maximal number of market-type traders, who can be continuously accepted by a limit order market with a spread $s$, is determined by the function,
\[
y(s)=\left\lfloor\frac{\log((1-\gamma)\overline{a})-\log s}{\log(1+\alpha)}+1\right\rfloor^+,
\]
where $\lfloor x\rfloor^+=\max\{\lfloor x\rfloor,0\}$. 
\end{pro}

There is also an exponential law in the relationship between $s$ and the market's capacity of accepting market-type traders $n$,
\begin{equation}
s\propto(1+\alpha)^{-n}\overline{a},\ \text{where}\ n\in[y(s),y(s)+1).
\end{equation}

The probability distribution function $g(s)$ of the appearance of market-type traders is negatively linear in $\log s$, say,
\[
g(s)=-k_3\log s+k_4,\ \text{if}\ s\in[\underline{s},(1-\gamma)\overline{a}],
\]
where $k_3>0$ and $k_4$ are again constants determined by $\underline{s},\overline{a},\alpha,\gamma$, and 
\[
g(s)=0,\ \text{if}\ s\notin[\underline{s},(1-\gamma)\overline{a}].
\]

Now we are prepared to consider the switching laws between different types of traders in any two consecutive periods. If $s\in[\underline{s},(1+\alpha)\underline{s})$, which means that the market has accepted too many limit-type traders, then the acceptable marginal trader in the next period will switch into the market-type for sure. If $s\in((1-\gamma)\overline{a},\underline{a}]$, then the acceptable marginal trader will switch from the market-type to the limit-type. Assume the switching probabilities as a function of the bid-ask spread are continuous over $[\underline{s},\overline{a}]$. The probability of switching from limit-type to market-type is $1$, if $s\in[\underline{s},(1+\alpha)\underline{s})$, and it is $0$, if $s\in((1-\gamma)\overline{a},\overline{a}]$. The probability of switching from limit-type to market-type is decreasing from $1$ to $0$ on the domain $[(1+\alpha)\underline{s},(1-\gamma)\overline{a}]$. On the other hand, the probability of switching from market-type to limit-type is increasing on the same domain. There also exist similar switching laws between buy-type (type BM and type BL) and sell-type (type SM and type SL), which will be studied in Section \ref{sec3.2}.

\section{Iterated Trading Process}\label{sec3}

\subsection{Sequential Trading}\label{sec3.1}

Define four linear functions mapping from $W$ into itself,
\[
f_i(\dd)=S_i\dd,\ i\in\{1,2,3,4\},
\]
where $S_1,S_2,S_3,S_4$ are $2\times 2$ matrices as defined in Section \ref{sec2.2}. There is an equivalent function $g_i(\ww)=\ww S_i'$, such that $g_i(\ww)=f'_i(\dd)$ if $\ww=\dd'$, where $S_i'$ is the transpose of $S_i$ for all $i$. Here, we again assume that the parameters $\alpha$ and $\beta$ in $S_1,S_2,S_3,S_4$ are equal, so $\beta$ will be denoted by $\alpha$ equivalently. Let $F$ be the collection of these four functions, namely, $F=\{f_1,f_2,f_3,f_4\}$. 

For each $i\in\{1,2,3,4\}$, and any given $\dd\in W$, we define a convex set,
\[
L_i(\dd)=\{\lambda\dd+(1-\lambda)S_i\dd:0\leq\lambda\leq1\}.
\]
So $L_i(\dd)$ is actually the line segment between $\dd$ and $S_i\dd$ in the $b$-$a$ plane. We can now give the definition of $f_i(\dd)$ more precisely, say,
\[
f_i(\dd)\in\bigg\{
\begin{aligned}
\{S_i\dd\}\quad\, &,\ \text{if}\ S_i\dd\in W\\
L_i(\dd)\cap\partial W &,\ \text{if}\ S_i\dd\notin W
\end{aligned}\ ,
\]
where $\partial W$ is the boundary of the closed domain $W$. Since both $\{S_i\dd\}$ and $L_i(\dd)\cap\partial W$ are singletons, $f_i(\dd)$ takes either the value of $S_i\dd$ or the unique element of $L_i(\dd)\cap\partial W$, hence it is well-defined as a function. By the definition of $f_i$, the domain $W$ then has an absorbing barrier, so that the dynamics of $f_i$ will be restricted within $W$. If a bid-ask pair touches $\partial W$ at $\dd$, then $f_i(\dd)=\dd$ as $L_i(\dd)\cap\partial W=\{\dd\}$, and hence it will be absorbed at $\partial W$. 

\begin{dfn}
If a limit order market with its bid-ask pair absorbed on $\partial W$ at a time $t$, the market after the time $t$ is said to be in a \markdef{crash}.
\end{dfn}

Consider a permutation function $\sigma:\{1,2,3,4\}\to\{\BL,\BM,\SL,\SM\}$, such that $\sigma(1)=\BL$, $\sigma(2)=\BM$, $\sigma(3)=\SL$, and $\sigma(4)=\SM$.

\begin{dfn}
The discrete dynamical system $(W,f_i)$ is called a \markdef{trading system} generated by a marginal trader of type $\sigma(i)$, where $i\in\{1,2,3,4\}$.
\end{dfn}

Each trading system $(W,f_i)$ produces a certain dynamics of bid-ask pairs in the domain $W$, where $i\in\{1,2,3,4\}$. Given any initial condition $\ww\in W$, the linear dynamics is quite clear, in which the bid-ask pair will eventually hit the point $\ww_i\in\partial W$ in the trading system $(W,f_i)$. Concretely, for any initial state $\ww=(b,a)\in W$, we have
\[
\ww_1=(a-\underline{s},a),\ \ww_2=(b,\overline{a}),\ \ww_3=(b,b+\underline{s}),\ \text{and}\ \ww_4=(0,a).
\] 
We want to state a generalized result based on this fact.

\begin{pro}\label{pro1}
If a same type marginal trader repeatedly comes to a market, then the market starting from any initial bid-ask pair in $W$ is unstable.
\end{pro}

\begin{proof}
Note that for any initial condition $\dd_t=(b_t,a_t)'$, after $n$ forward periods with the marginal trader of the same type $\sigma(1)$, the bid-ask pair will be $S_1^n\dd_t$ that converges to $(a_t,a_t)'$ if $n$ is sufficiently large, where
\[
S_1^n=
\begin{pmatrix}
\frac{1}{(1+\alpha)^n} & \frac{\alpha\sum_{i=0}^{n-1}(1+\alpha)^i}{(1+\alpha)^n}\\
0 & 1
\end{pmatrix}
=
\begin{pmatrix}
\frac{1}{(1+\alpha)^n} & 1-\frac{1}{(1+\alpha)^n}\\
0 & 1
\end{pmatrix}
\to
\begin{pmatrix}
0 & 1\\
0 & 1
\end{pmatrix}.
\]
But in the trading system $(W,f_1)$, the bid-ask pair should always stay in $W$, and hence the last bid-ask pair remains in the system is $(a_t-\underline{s},a_t)\in\partial W$. Thus the trading system monotonically moves to a crash, and the market will then stop.

The same happens to the trading system $(W,f_i)$, where $i\in\{2,3,4\}$, as
\[
S_2^n=
\begin{pmatrix}
1 & 0\\
1-(1+\alpha)^n & (1+\alpha)^n
\end{pmatrix}
\to
\begin{pmatrix}
1 & 0\\
-\infty & \infty
\end{pmatrix},
\]
which implies $S_2^n\dd_t\to(b_t,\infty)'$,
\[
S_3^n=
\begin{pmatrix}
1 & 0\\
1-\frac{1}{(1+\alpha)^n} & \frac{1}{(1+\alpha)^n}\\
\end{pmatrix}
\to
\begin{pmatrix}
1 & 0\\
1 & 0
\end{pmatrix},
\]
which implies $S_3^n\dd_t\to(b_t,b_t)'$, and finally
\[
S_4^n=
\begin{pmatrix}
(1+\alpha)^n & 1-(1+\alpha)^n\\
0 & 1
\end{pmatrix}
\to
\begin{pmatrix}
\infty & -\infty\\
0 & 1
\end{pmatrix},
\] 
which implies $S_4^n\dd_t\to(-\infty,a_t)'$. Thus $(W,f_2)$, $(W,f_3)$, and $(W,f_4)$ will be in a crash after achieving $(b_t,\overline{a})$, $(b_t,b_t+\underline{s})$, and $(0,a_t)$, respectively.
\end{proof}

In Proposition \ref{pro1}, we actually consider a sequence of traders with a constant type, say $\{q,q,\dotsc\}$, where $q\in\{\sigma(1),\sigma(2),\sigma(3),\sigma(4)\}$. The marginal trader of type $q$ comes to the trading system $(W,f_i)$, where $i=\sigma^{-1}(q)$ is determined by the permutation scheme $\sigma$. We state that the trading system determined by $\{q,q,\dotsc\}$ is not stable, in the sense that it will crash at the boundary of the domain $W$.

Now we define a general sequence of marginal traders with different types starting from time $t$ as $\{q_t,q_{t+1},\dotsc\}$, where $q_{t+i}\in\{\sigma(1),\sigma(2),\sigma(3),\sigma(4)\}$ for all $i\in\Z_+$. Given a trader's type $q=\sigma(i)$ at the time $t$, the atomic trading scheme at $t$ will be determined by $f_i$, for all $i\in\{1,2,3,4\}$.

\begin{dfn}
The iterated function system $(W,F)$ is called an \markdef{iterated trading system} generated by a sequence of traders, such that the atomic trading scheme is $f_i$ if the trader in the sequence is of type $\sigma(i)$. Moreover, we denote the iterated trading system as a triplet $(W,F,\sigma)$.
\end{dfn}

Sometimes we also refer to $(W,F,\sigma)$ as a \markdef{dynamical trading system}. Recall that $(W,f_i)$ is a trading system, so $(W,\{f_i\},\sigma)$ is trivially an iterated trading system, which is exactly equivalent with the trading system $(W,f_i)$ for all $i\in\{1,2,3,4\}$. Note that each $f_i\in F$ has a common parameter $\alpha$, so the iterated trading system $(W,F,\sigma)$ also depends on $\alpha$.

First of all, we are interested in identifying stable components, which do not change any bid-ask pair in the iterated trading system $(W,F,\sigma)$. Notice that $S_1=S_4^{-1}$, and $S_2=S_3^{-1}$ for all $\alpha\in(0,1)$, so $S_1S_4=S_2S_3=I$, where $I$ is the identity matrix of order $2$, and hence
\[
f_1\circ f_4=f_4\circ f_1=\id_W,\ \text{and}\ f_2\circ f_3=f_3\circ f_2=\id_W,
\]
where $\id_W$ is the identity function on $W$. Therefore, $\{\sigma(1),\sigma(4)\}$, $\{\sigma(4),\sigma(1)\}$, $\{\sigma(2),\sigma(3)\}$, and $\{\sigma(3),\sigma(2)\}$ are all stable components in any sequence of marginal traders for all $\alpha\in(0,1)$. In general, such stable components are called periodic blocks.

\begin{dfn}
A \markdef{periodic block} is a consecutive component in a given sequence of traders, which does not change any bid-ask pair in a specific iterated trading system.
\end{dfn}

Note that any combination of periodic blocks is again a periodic block. For instance, $\{\sigma(1),\sigma(4),\sigma(2),\sigma(3)\}$ is a periodic block for all $\alpha\in(0,1)$, as $\{\sigma(1),\sigma(4)\}$ and $\{\sigma(2),\sigma(3)\}$ are general periodic blocks. So we need to catch the kernel of a periodic block, such that it is invariant, namely, its kernel should be itself.

\begin{dfn}
A periodic block is \markdef{minimal}, if it has no proper subtuple that is again a periodic block.
\end{dfn}

Any periodic block can be reduced into a series of minimal ones. Note that a periodic block $C$ is either minimal or not minimal. If $C$ is minimal, it is equivalent with itself. If $C$ is not minimal, we can always find a proper subtuple $C'\subset C$ such that $C'$ and $C\setminus C'$ are periodic blocks. We can eventually have a series of minimal periodic blocks by applying this partition process recursively. 

\begin{emp}
If $\alpha\in(0,1)$, the periodic block $\{\sigma(1),\sigma(2),\sigma(3),\sigma(4)\}$ has two minimal periodic blocks, namely, $\{\sigma(2),\sigma(3)\}$ and $\{\sigma(1),\sigma(4)\}$, and the periodic block
\[
\{\sigma(4),\sigma(1),\sigma(1),\sigma(4),\sigma(4),\sigma(1)\}
\]
also has two minimal periodic blocks, namely, $\{\sigma(4),\sigma(1)\}$ and $\{\sigma(1),\sigma(4)\}$. 

If $\alpha=1/2$,
\[
\{\sigma(2),\sigma(1),\sigma(2),\sigma(1),\sigma(3),\sigma(4),\sigma(3),\sigma(4),\sigma(3),\sigma(4)\}
\]
is a periodic block, and it is minimal.

If $\alpha=1/3$,
\[
\{\sigma(1),\sigma(2),\sigma(1),\sigma(2),\sigma(1),\sigma(2),\sigma(1),\sigma(2),\sigma(4),\sigma(3),\sigma(4),\sigma(3),\sigma(4),\sigma(3)\}
\]
is a minimal periodic block.
\end{emp}

\begin{lem}\label{lem1}
The number of marginal traders in any minimal periodic block is finite and even.
\end{lem}

\begin{proof}
Consider a minimal periodic block $C$, and assume the number of traders in $C$ is infinite. Then $C$ must pass infinite bid-ask pairs. If not, we suppose $C$ passes finite bid-ask pairs. Since the number of traders in $C$ is infinite, there must exist a closed route, such that a related subset of $C$ is a periodic block, which contradicts that $C$ is minimal. 

Assume the initial condition of $C$ is $\ww\in W$ with a bid-ask spread $s$. Note that $s$ can be updated into either $(1+\alpha)s$ or $s/(1+\alpha)$, so any bid-ask pair on the trajectory will have a spread in the set
\[
S_\ww=\{(1+\alpha)^i s:-N_1\leq i\leq N_2\ \text{and}\ i\in\Z\},
\]
where $N_1,N_2\in\Z_+$ are finite, since $W$ is bounded. Let the bid-ask pair $\ww_r$ be the first state with a spread $r$ on the trajectory starting from $\ww$, for all $r\in S_\ww$, where $\ww_r=\ww$, if $r=s$. $\ww_r$ will be updated to $\ww_r\pm(\alpha r/(1+\alpha),\alpha r/(1+\alpha))$ by the block $\{\sigma(1),\sigma(2)\}$ or $\{\sigma(3),\sigma(4)\}$, and to $\ww_r\pm(\alpha r,\alpha r)$ by the block $\{\sigma(2),\sigma(1)\}$ or $\{\sigma(4),\sigma(3)\}$. So at the constant-spread line $a-b=r$, all the possible states on the trajectory have the form,
\[
\ww_r+k_1(\alpha r,\alpha r)+k_2\left(\frac{\alpha}{1+\alpha} r,\frac{\alpha}{1+\alpha} r\right)=\ww_r+\left(k_1\alpha+\frac{k_2\alpha}{1+\alpha},k_1\alpha+\frac{k_2\alpha}{1+\alpha}\right)r,
\]
where $k_1,k_2\in\Z$ and they are finite, as $W$ is bounded. Since $\alpha,k_1,k_2$ are finite, all the possible states with a given spread $r$ on the trajectory are finite, and hence all the states in $W$ starting from $\ww$ are finite. So $C$ can not pass infinite bid-ask pairs, which implies the number of marginal traders in $C$ must be finite.

Suppose $C$ has $2n+1$ traders, where $n\in\Z_+$, and assume it will pass $m$ different bid-ask pairs, the collection of which is denoted by the set $P$. By Proposition \ref{pro1}, any trader of type $\sigma(i)$ will definitely update the bid-ask pair in a trading system $(W,f_i)$, so $1<m<\infty$. Let $P$ be the set of nodes in a graph, so any trader in $C$ will link two different nodes in $P$. Since there are $2n+1$ traders, we have $2n+1$ links in this graph. But if there exists a directed circle, such that the bid-ask pair after this block will not be changed, then the number of links of any node in $P$ should be even. So the total links in this graph should be even, which contradicts that the number of links is $2n+1$. Therefore, a block with $2n+1$ traders can not be periodic, which completes the proof.
\end{proof}

Note that we can have an equivalent reduced sequence of traders by deleting (minimal) periodic blocks iteratively from any sequence of traders, as we just delete some closed routes of bid-ask pairs, which will not change the dynamics of an iterated trading system as a whole. 

\begin{dfn}
A sequence of marginal traders is \markdef{irreducible}, if it does not contain any minimal periodic block.
\end{dfn}

\begin{pro}\label{pro3}
A market accepts any irreducible sequence of traders is unstable.
\end{pro}

\begin{proof}
Note that any market functioning for infinite periods must contain several minimal periodic blocks, otherwise it is a minimal periodic block with infinite traders, which contradicts Lemma \ref{lem1}. Since any irreducible sequence of marginal traders contains no minimal periodic block, the number of marginal traders in any irreducible sequence must be finite, otherwise we have a market functioning with infinite periods has no minimal periodic block. If the number of marginal traders in a sequence is finite, then the market must function only for finite periods. So the bid-ask pair in the market must be absorbed on $\partial W$, and hence the market will result in a crash.
\end{proof}

\begin{rmk}
It is clear that any market functioning for infinite periods will never accept an irreducible sequence of marginal traders. Or we can say the sequence of marginal traders in a stable market should be infinite and reducible, so that we can always find some minimal periodic blocks lasting for finite periods in the market. 

A similar concept to periodic block is the well-known ``hedging'' in finance. Our result suggests that the periodic blocks as hedging units in a limit order market are necessary for its dynamic stability.
\end{rmk}

\subsection{Stochastic Trading}\label{sec3.2}

In Section \ref{sec3.1}, we study the iterated trading process in a combinatorial way. In fact, we consider all the possible enumerations for a sequence of marginal traders, where the type of any marginal trader in a sequence belongs to the set
\[
\Sigma_4=\{\sigma(1),\sigma(2),\sigma(3),\sigma(4)\}=\{\BL,\BM,\SL,\SM\}.
\]
So all the sequences of marginal traders form the space $\Sigma_4^\infty$. We establish some results on the relationship between the stability of a limit order market and certain subsets of $\Sigma_4^\infty$. We find two general categories of sequences in $\Sigma_4^\infty$ are unstable in a limit order market with any initial state in $W$. Namely, the sequence of marginal traders with a constant type, say $\{q,q,\dotsc\}$, where $q\in\Sigma_4$, as stated in Proposition \ref{pro1}, and any irreducible sequence that contains no minimal periodic block, as stated in Proposition \ref{pro3}.

In this section, we will take a different perspective to study the iterated trading process in the limit order market. We assume there exists a certain probability measure on the space $\Sigma_4^\infty$, so the dynamics of bid-ask pairs in the iterated trading system $(W,F,\sigma)$ will become random. Not surprisingly, the related limit order market is stochastically stable, since the random trajectory in $(W,F,\sigma)$, again controlled by the switching rules, will not be absorbed on $\partial W$ almost surely. So the stochastic dynamics of bid-ask pairs in the limit order market will not generate crashes almost surely. 

To construct a reasonable probability measure on $\Sigma_4^\infty$, we will rely again on the ``gravitational force,'' which determines the switching rules between different types of marginal traders, as discussed in Section \ref{sec2.3}. At first, we assume the general probability measure on $\Sigma_4^\infty$ can be represented by a same stationary probability measure on $\Sigma_4$ at each time. Note that the probability measure on $\Sigma_4$ is a vector in $[0,1]^4$, but it takes different values at different bid-ask pairs according to the switching rules in the market. So the probability measure on $\Sigma_4$ can be thought of as a function mapping from $W$ to $[0,1]^4$. Through this construction, we can rebuild the relationship between the probability measure on $\Sigma_4$ and the time, since at each time, any marginal trader in the market must be associated with a certain bid-ask pair in $W$.

Without loss of generality, we assume $\alpha=\beta$ once again. Consider an arbitrary initial bid-ask pair $\ww=(b,a)\in W$, and suppose its bid-ask spread and mid-price are denoted by $s=a-b$ and $m=(b+a)/2$ as before. If a marginal trader of type $\sigma(1)$ comes to the market, the bid-ask pair in the next period will be $(b_+,a)$, where $b_+>b$. If the marginal trader is of type $\sigma(2)$, it will be $(b,a_+)$, where $a_+>a$. If the marginal trader is of type $\sigma(3)$, it will be $(b,a_-)$, where $a_-<a$. Finally, if the marginal trader is of type $\sigma(4)$, it will be $(b_-,a)$, where $b_-<b$. 

Through some computations using the results in Section \ref{sec2.2}, we get
\[
b_+-b=a-a_-=\frac{\alpha}{1+\alpha}s,\ \text{and}\ a-b_+=a_--b=\frac{s}{1+\alpha},
\]
and similarly, we can also obtain
\[
a_+-a=b-b_-=\alpha s,\ \text{and}\ a_+-b=a-b_-=(1+\alpha)s.
\]
So the types $\sigma(2)$ and $\sigma(4)$ will cause a larger bid-ask spread than $s$, say $(1+\alpha)s$, while the types $\sigma(1)$ and $\sigma(3)$ will cause a smaller bid-ask spread than $s$, say $s/(1+\alpha)$. In fact, $\sigma(2)$ and $\sigma(4)$ are market-type, and $\sigma(1)$ and $\sigma(3)$ are limit-type. 

Note that 
\[
a+b_+=2m+\frac{\alpha}{1+\alpha}s,\ \text{and}\ a_++b=2m+\alpha s,
\]
as $a_+-a>b_+-b>0$, and 
\[
a_-+b=2m-\frac{\alpha}{1+\alpha}s,\ \text{and}\ a+b_-=2m-\alpha s,
\]
as $b_--b<a_--a<0$. The types $\sigma(1)$ and $\sigma(2)$ will cause a new mid-price larger than $m$, while the types $\sigma(3)$ and $\sigma(4)$ will cause a lower mid-price than $m$. This similarity also suggests that we may consider $\sigma(1)$ and $\sigma(2)$ in combination as buy-type, and $\sigma(3)$ and $\sigma(4)$ jointly as sell-type.

Our construction of the probability measure on $\Sigma_4$ starts from the investigation in the switching laws between limit-type and market-type traders, and between buy-type and sell-type traders. Suppose the capacities of accepting different types of traders in a limit order market are balanced. That's to say, if a limit order market has accepted too many market-type traders, it will be less likely to accept an additional market-type trader, but more likely to accept an additional limit-type trader, and vice versa. If the market has accepted too many buy-type traders, the probability of a new sell-type trader will be extremely high, and the probability of an additional buy-type trader will be very low, and vice versa. Note that the switching between $\sigma(i)$ and $\sigma(j)$ is composed by the limit-market switching and the buy-sell switching, where $i,j\in\{1,2,3,4\}$, so we can obtain the overall switching laws and induced probability measures in $\Sigma_4$ from the buy-sell and limit-market switching laws.

Let the probability measure on $\Sigma_4$ be $\pi:W\to[0,1]^4$,
\[
\pi(\ww)=(\pi_1(\ww),\pi_2(\ww),\pi_3(\ww),\pi_4(\ww))',\ \text{for\ all}\ \ww\in W,
\]
such that $\sum_{i=1}^4\pi_i(\ww)=1$, where $\pi_i(\ww)$ is the probability that a type-$\sigma(i)$ marginal trader comes to the limit order market at the state $\ww$, and $\pi_i$ is also a function mapping $W$ to $[0,1]$, for all $i\in\{1,2,3,4\}$. Let the probability function on the domain $\Sigma_4$ at the state $\ww$ be $\kappa_\ww:\Sigma_4\to[0,1]$, so for all $\sigma(i)\in\Sigma_4$,
\begin{equation}
\kappa_\ww(\sigma(i))=\pi_i(\ww),\ \text{for\ all}\ \ww\in W.
\end{equation}

In addition, the probability that a limit-type trader comes to the market at the state $\ww$ is denoted by $\pi_L(\ww)=\pi_1(\ww)+\pi_3(\ww)$, and the probability of a new market-type trader is denoted by $\pi_M(\ww)=\pi_2(\ww)+\pi_4(\ww)$. The probability of a new buy-type trader is denoted by $\pi_B(\ww)=\pi_1(\ww)+\pi_2(\ww)$, and the probability of a new sell-type trader is denoted by $\pi_S(\ww)=\pi_3(\ww)+\pi_4(\ww)$. Evidently, we have
\[
\pi_L(\ww)+\pi_M(\ww)=\pi_B(\ww)+\pi_S(\ww)=1,\ \text{for\ all}\ \ww\in W.
\]
So the vector $(\pi_L(\ww),\pi_M(\ww),\pi_B(\ww),\pi_S(\ww))'$ can be uniquely fixed in a space with dimension two, and hence it is equivalent with $(\pi_L(\ww),\pi_B(\ww))'$. Recall that $\sum_{i=1}^4\pi_i(\ww)=1$, so $(\pi_1(\ww),\pi_2(\ww),\pi_3(\ww),\pi_4(\ww))'$ can be expressed uniquely in a space with dimension three, and hence it is equivalent with $(\pi_1(\ww),\pi_2(\ww),\pi_3(\ww))'$. 

There is a unique $(\pi_L,\pi_B)$ linearly derived from $(\pi_1,\pi_2,\pi_3)$, since
\[
\pi_L=\pi_1+\pi_3,\ \text{and}\ \pi_B=\pi_1+\pi_2.
\]
However, there does not exist a unique $(\pi_1,\pi_2,\pi_3)$ corresponding to $(\pi_L,\pi_B)$, unless we make some additional assumptions. One typical possibility is that we assume $\pi_L$ and $\pi_B$ are linearly independent, so that we can get
\[
\pi_1=\pi_L\pi_B,\ \pi_2=(1-\pi_L)\pi_B,\ \text{and}\ \pi_3=\pi_L(1-\pi_B).
\]

According to Proposition \ref{pro2}, we know the market's capacity of accepting limit-type traders is positively related with the natural logarithm of the bid-ask spread. So we assume $\pi_L(\ww)$ is monotonically increasing \wrt\ $\log s$, and thus $\pi_M(\ww)=1-\pi_L(\ww)$ is monotonically decreasing \wrt\ $\log s$, where $s$ is the spread of $\ww$. Moreover, if $s\in[\underline{s},(1+\alpha)\underline{s})$, or $\ww$ belongs to the region
\[
W_M=\{(b,a):\underline{s}\leq a-b<(1+\alpha)\underline{s}\},
\]
$\pi_M(\ww)=1$ and $\pi_L(\ww)=0$. If $s\in((1-\gamma)\overline{a},\overline{a}]$, or $\ww$ belongs to the region
\[
W_L=\{(b,a):(1-\gamma)\overline{a}<a-b\leq\overline{a}\},
\]
$\pi_L(\ww)=1$ and $\pi_M(\ww)=0$. Recall that $0<\alpha<1$ and $\alpha/(1+\alpha)\leq\gamma<1$ as defined in Section \ref{sec2.3}.

Similarly, we assume the market's capacity of accepting buy-type traders is negatively related with the natural logarithm of the mid-price. Thus we assume $\pi_B(\ww)$ is monotonically decreasing \wrt\ $\log m$, and hence $\pi_S(\ww)=1-\pi_B(\ww)$ is monotonically increasing \wrt\ $\log m$, where $m$ is the mid-price of $\ww$. Moreover, if $m$ is sufficiently low, namely, $\ww$ belongs to the region
\[
W_B=\{(b,a):\underline{s}\leq b+a<(1+\delta)\underline{s}\}, 
\]
where $\delta>0$ is a constant, $\pi_B(\ww)=1$ and $\pi_S(\ww)=0$. If $m$ is sufficiently high, namely, $\ww$ belongs to the region
\[
W_S=\{(b,a):(1-\epsilon)(2\overline{a}-\underline{s})<b+a\leq 2\overline{a}-\underline{s}\},
\]
$\pi_S(\ww)=1$ and $\pi_B(\ww)=0$, where $\delta/(1+\delta)\leq\epsilon<1$, since $0<\epsilon<1$ and $(1-\epsilon)(1+\delta)\leq 1$. 

\begin{dfn}
The \markdef{buffering region} of $W$ is the largest nonclosed subset $H\subset W$ with the property that $\prod_{x\in\{L,M,B,S\}}\pi_x(\ww)=0$ for all $\ww\in H$.
\end{dfn}

At any state $\ww\in H$, there exists at least an $x\in\{L,M,B,S\}$ such that $\pi_x(\ww)=0$. Since $\pi_L+\pi_M=\pi_B+\pi_S=1$, there also exists at least a $y\in\{L,M,B,S\}$ such that $\pi_y(\ww)=1$ at the state $\ww$. So we can have at most two elements, namely, $x_1\in\{L,M\}$ and $x_2\in\{B,S\}$, such that $\pi_{x_1}(\ww)=\pi_{x_2}(\ww)=0$, and $\pi_y(\ww)=1$ for $y\in\{L,M,B,S\}\setminus\{x_1,x_2\}$.

\begin{dfn}
The \markdef{kernel region} of $W$ is the largest closed subset $K\subseteq W$ with the property that $\pi_x(\ww)\neq 0$ for all $x\in\{L,M,B,S\}$ and for all $\ww\in K$.
\end{dfn}

In general, we have $K=W\setminus H$, and $K\cap H=\emptyset$. Thus, we have a bipartition of the domain $W$, namely $K\cup H=W$ and $K\cap H=\emptyset$. Since $K$ is defined to be closed, and $H$ is defined to be nonclosed, $H$ may be empty, but $H\neq W$, so $K$ is always nonempty. If $K=W$, then $H=\emptyset$. If $K=\{\ww\}$, where $\ww\in W$, then $H=W\setminus\{\ww\}$. 

In our linear setting, we have $\pi_L(\ww)=0$ for all $\ww\in W_M$, $\pi_M(\ww)=0$ for all $\ww\in W_L$, $\pi_B(\ww)=0$ for all $\ww\in W_S$, and $\pi_S(\ww)=0$ for all $\ww\in W_B$. Thus
\[
H=W_L\cup W_M\cup W_B\cup W_S,
\]
and $K=W\setminus H$, where $W_L\cap W_M=\emptyset$, $W_B\cap W_S=\emptyset$. Note that $K$ is closed, $H$ is not closed, but $H\cup\partial K$ is also closed, where $\partial K$ is the boundary of the kernel region $K$.

Recalling that $\ww=(b,a)\in W$, we define a spread function $s:W\to\R$, 
\[
s(\ww)=a-b,\ \text{for\ all}\ \ww\in W,
\]
and a mid-price function $m:W\to\R$, 
\[
m(\ww)=(b+a)/2,\ \text{for\ all}\ \ww\in W.
\]
Thus the bid-ask spread at $\ww$ is now denoted specifically by $s(\ww)$ instead of the general $s$, and the mid-price at $\ww$ is denoted by $m(\ww)$ instead of $m$.

\begin{dfn}
The \markdef{$\boldsymbol{s}$-range} of a subset $R\subseteq W$ is defined as
\[
r_s(R)=\sup_{\ww\in R}s(\ww)-\inf_{\vv\in R}s(\vv),
\]
and the \markdef{$\boldsymbol{m}$-range} of $R$ is defined as
\[
r_m(R)=\sup_{\ww\in R}m(\ww)-\inf_{\vv\in R}m(\vv).
\]
\end{dfn}

\begin{pro}\label{pro4}
The dynamical trading system $(W,F,\sigma)$ is stochastically stable, and its trajectory of bid-ask pairs will stay within $K$ almost surely, if (i) the buffering region $H$ is nonempty, (ii) $\min\{r_s(K),r_m(K)\}>\alpha(1+\alpha)(2+\alpha)\underline{s}$, and (iii) $\pi$ is strictly monotonic on $K$.
\end{pro}

\begin{proof}
For any trajectory starts from a bid-ask pair $\ww\in W$, all the possible states in $W$ are finite, as shown in the proof of Lemma \ref{lem1}. We denote the set of all the possible states for any initial state $\ww$ by a corresponding lattice $\Lambda(\ww)$. Let the neighborhood of any $\vv\in\Lambda(\ww)$ be
\[
N(\vv)=\{\vv S_1',\vv S_2',\vv S_3',\vv S_4'\}\cap W,
\]
where $S_i'$ is the transpose of the matrix $S_i$ for all $i\in\{1,2,3,4\}$. Since $\Lambda(\ww)$ is connected in $W$, there exist states $\vv\in\Lambda(\ww)$ near $\partial K$, such that $N(\vv)\cap H\neq\emptyset$ and $N(\vv)\cap K\neq\emptyset$. 

Note that the spread $s(\vv)$ of the state $\vv$ can be updated to $s(\vv)/(1+\alpha)$ and $(1+\alpha)s(\vv)$, and the mid-price $m(\vv)$ of the same state $\vv$ can be updated to maximally $m(\vv)+\alpha s(\vv)/2$ and minimally $m(\vv)-\alpha s(\vv)/2$, so 
\[
r_s(N(\vv))=\frac{\alpha(2+\alpha)}{1+\alpha}s(\vv),\ \text{and}\ r_m(N(\vv))=\alpha s(\vv).
\]
Suppose $s(\vv)/(1+\alpha)\geq(1+\alpha)\underline{s}$, where $(1+\alpha)\underline{s}$ is the lower bound of the spread in $K$, so we have
\[
r_s(N(\vv))\geq\alpha(1+\alpha)(2+\alpha)\underline{s},\ \text{and}\ r_m(N(\vv))\geq\alpha(1+\alpha)^2\underline{s}.
\]
By the condition (ii), $r_s(K)>\inf_{\vv\in K}r_s(N(\vv))$, and $r_m(K)>\inf_{\vv\in K}r_m(N(\vv))$. So for any initial state $\ww\in W$, there exists at least a $\vv\in \Lambda(\ww)\cap K$ such that $N(\vv)\subset K$.

Note that $W=K\cup H$, and both $K$ and $H$ are nonempty, so both $H$ and $K$ are proper subsets of $W$. Suppose $\ww\in H$. We know $H=W_L\cup W_M\cup W_B\cup W_S$, so there exists at least an $x\in\{L,M,B,S\}$ such that $\ww\in W_x$. Note that $\pi_x(\vv)=1$ for all $\vv\in W_x$, thus $\ww$ will move towards $K$ along a continuous flow in $\Lambda(\ww)\cap H$. Since the number of states in $\Lambda(\ww)\cap H$ is finite, $\ww$ will move into the closed kernel region $K$ after finite periods, say $k(\ww)<\infty$.

Let the probability that bid-ask pairs stay in $H$ with an initial state $\ww\in W$ is $p(\ww)$. Assume $p(\ww)=0$ for all $\ww\in K$. So the market will be stochastically stable within $K$ once $\ww\in K$, and the market will then function for infinite periods. If $\ww\in H$,
\[
p(\ww)=\lim_{T\to\infty}1\times\frac{k(\ww)}{T}+0\times\left(1-\frac{k(\ww)}{T}\right)=\lim_{T\to\infty}\frac{k(\ww)}{T}=0.
\]
We thereof only need to show $p(\ww)=0$ when $\ww\in K$.

Suppose $\ww\in K$. We have two possibilities, namely, $N(\ww)\subset K$, and $N(\ww)\cap H\neq\emptyset$. If $N(\ww)\cap H\neq\emptyset$, then $N(\ww)\cap K\neq\emptyset$, otherwise $\ww\in H$. $\pi_x:W\to[0,1]$ is strictly monotonic on $K$ for all $x\in\{L,M,B,S\}$, as $\pi:W\to[0,1]^4$ is strictly monotonic on $K$, and $\pi_x$ is equal to the sum of two distinct functions taken from $\{\pi_1,\pi_2,\pi_3,\pi_4\}$. So there exist some $x\in\{L,M,B,S\}$, such that $\pi_x(\vv)=0$ for all $\vv\in N(\ww)\cap H$, and $\pi_x(\vv)=\varepsilon_x\in[0,1]$ for all $\vv\in N(\ww)\cap K$. Note that $\varepsilon_x=1$ if only $\vv\in N(\ww)\cap\partial K$, but then it will not move into $N(\ww)\cap H$ in the next period. So we only need to consider the case that $\vv\notin\partial K$, hence $\varepsilon_x\neq 1$ for all $x\in\{L,M,B,S\}$. 

If $\ww$ moves to $\vv\in N(\ww)\cap H$ with a probability $\varepsilon_x$, it will return back to $\vv'\in N(\vv)\cap K$ with probability $1$ in the next period, where $N(\vv')\cap H\neq\emptyset$ as $\vv\in N(\vv')$. If $\ww$ moves to $\vv\in N(\ww)\cap K$ with a probability $1-\varepsilon_x$, it can stay within $\Lambda(\ww)\cap K$ with $h(\ww)$ continuous periods, and then move into a state $\vv'$ such that $N(\vv')\cap H\neq\emptyset$. Recall that
\[
\{\vv\in\Lambda(\ww)\cap K:N(\vv)\subset K\}\neq\emptyset,\ \text{for\ all}\ \ww\in K,
\]
so $h(\ww)\geq 0$. We obtain
\[
p(\ww)=\lim_{T\to\infty}\varepsilon_x p(\ww)\times\left(1-\frac{2}{T}\right)+(1-\varepsilon_x) p(\ww)\times\left(1-\frac{h(\ww)+1}{T}\right),
\]
where $0\leq h(\ww)\leq T-1$. When $h(\ww)=T-1$,
\[
p(\ww)=\lim_{T\to\infty}\varepsilon_x p(\ww)\times\left(1-\frac{2}{T}\right)=\varepsilon_x p(\ww),
\]
which generates $(1-\varepsilon_x)p(\ww)= 0$. Since $1-\varepsilon_x\neq 0$, $p(\ww)=0$ for all $\ww\in K$ such that $N(\ww)\cap H\neq\emptyset$.

If $\ww\in K$ and $N(\ww)\subset K$, its trajectory can either achieve a state $\ww'\in\Lambda(\ww)\cap K$ such that $N(\ww')\cap H\neq\emptyset$ after $j(\ww)$ periods, where $j(\ww)\geq 1$, or never move to such a state $\ww'$ and thus it stay within $K$ for ever. Note that $p(\ww')=0$, if $\ww'\in K$ and $N(\ww')\cap H\neq\emptyset$. So $p(\ww)\leq p(\ww')=0$, but $p(\ww)\geq 0$, hence $p(\ww)=0$ for all $\ww\in K$ such that $N(\ww)\subset K$.

As a result, $p(\ww)=0$ if $\ww\in K$, and thus $p(\ww)=0$ for all $\ww\in W$. So the dynamical trading system is stable within $K$ almost surely.
\end{proof}

Since $(b_t,a_t)$ will stay within $K$ almost surely, the random trajectories of $b_t$ and $a_t$ will also stay in bounded intervals. Note that
\[
(1+\alpha)\underline{s}\leq s_t \leq r_s(K)+(1+\alpha)\underline{s},
\]
and
\[
(1+\delta)\underline{s}/2\leq  m_t \leq r_m(K)+(1+\delta)\underline{s}/2,
\]
for all $t\in\Z$, so
\begin{equation}
a_t\leq r_m(K)+r_s(K)/2+(2+\alpha+\delta)\underline{s}/2<r_m(K)+r_s(K)/2+2\underline{s},
\end{equation}
and
\begin{equation}
b_t\leq r_m(K)+(2+\alpha+\delta)\underline{s}/2<r_m(K)+2\underline{s},
\end{equation}
where $\alpha,\delta\in(0,1)$, so $(2+\alpha+\delta)/2<2$.

The upper bounds of $b_t$ and $a_t$ have interesting implications on the roles of $s$-range and $m$-range in the limit order market. The upper bound of the best bid $b_t$ in the market is solely determined by the $m$-range of the kernel region $K$ of $W$, rather than any property of the whole market. Similarly, the upper bound of the best ask $a_t$ is determined by the $m$-range and $s$-range of $K$, and uncorrelated with the property of the buffering region $H$. Notice that the lower bounds of $b_t$ and $a_t$ are close to $2\underline{s}$, so the bid-range in the market is approximately equal to $r_m(K)$, and the ask-range is roughly $r_m(K)+r_s(K)/2$. Evidently, the volatility in the ask side will be greater than the volatility in the bid side of the limit order book.

At each time $t\in\Z$, $(s_t,m_t)$ and $(b_t,a_t)$ are uniquely determined by each other, as
\[
\begin{pmatrix}
s_t\\
m_t
\end{pmatrix}=
\begin{pmatrix}
-1 & 1\\
\frac{1}{2} & \frac{1}{2}
\end{pmatrix}
\begin{pmatrix}
b_t\\
a_t
\end{pmatrix},
\]
where the $2\times 2$ transformation matrix is singular. So the random trajectory of $(b_t,a_t)$ remains in the kernel region $K$ is equivalent with the fact that the random trajectory of the corresponding $(s_t,m_t)$ will remain in a region $K'$, where $K'$ is determined by the $s$-range and $m$-range, and it is linearly transformed from $K$.

The stochastic process $\{s_t,\ t\in\Z\}$ has a binomial property, namely, at each $t\in\Z$, $s_t$ can be updated to $s_{t+1}=s_t/(1+\alpha)$ with a probability $\pi_L$, and to $s_t(1+\alpha)$ with a probability $1-\pi_L$, where $\pi_L$ is a smooth function of $\log s_t$, and hence also a function of $t$. However, the stochastic process $\{m_t,\ t\in\Z\}$ only has the Markov property, as $m_{t+1}-m_t$ is completely determined by $s_t$ and a stationary discrete random variable. At last, note that $b_t=a_t-s_t$, so the time series $\{b_t,\ t\in\Z\}$ and $\{a_t,\ t\in\Z\}$ are linearly correlated. 

\section{Controlled Trading System}

In this section, we assume again $K\neq\emptyset$, but either its $s$-range or its $m$-range is less than $\alpha(1+\alpha)(2+\alpha)\underline{s}$. So the condition (ii) in Proposition \ref{pro4} is no longer satisfied. We want to check whether the random trajectories of bid-ask pairs can maintain the property of stochastic stability within certain domains. 

Let $U_1=W_L\cup W_M$. Since $K\neq\emptyset$, $W_L\cap W_M=\emptyset$ and $W\setminus U_1\neq\emptyset$. Define $U_2=W\setminus U_1$, so $W=U_1\cup U_2$, and $U_1\cap U_2=\emptyset$. Similarly, let $V_1=W_B\cup W_S$, again $W\setminus V_1\neq\emptyset$. Define $V_2=W\setminus V_1$, so $W=V_1\cup V_2$, and $V_1\cap V_2=\emptyset$. Note that $K=U_2\cap V_2$ and $H=U_1\cup V_1$.

\subsection{Controlled Spread Dynamics}

By the condition $W_L\cap W_M=\emptyset$, we have $r_s(U_2)\geq 0$, so $(1+\alpha)\underline{s}\leq (1-\gamma)\overline{a}$, or
\[
\underline{s}/\overline{a}\leq\frac{1-\gamma}{1+\alpha}.
\]
At the same time, we assume that the $s$-range of $U_2$ is sufficiently small, namely, $r_s(U_2)<\alpha(1+\alpha)\underline{s}$, so
\[
(1-\gamma)\overline{a}-(1+\alpha)\underline{s}<\alpha(1+\alpha)\underline{s},
\]
which implies that $\underline{s}/\overline{a}$ has a lower bound,
\[
\underline{s}/\overline{a}>\frac{1-\gamma}{(1+\alpha)^2}.
\]
Intuitively, the upper bound of $s$-range of $U_2$ gives a sufficient condition that any type of marginal trader will definitely update any $\ww\in U_2$ to some $\ww'\in U_1$. 

In sum, if $0\leq r_s(U_2)<\alpha(1+\alpha)\underline{s}$, the domain $W$ will have the following property,
\begin{equation}
\frac{1-\gamma}{(1+\alpha)^2}<\underline{s}/\overline{a}\leq\frac{1-\gamma}{1+\alpha}.
\end{equation}
Once the above inequality is satisfied by the domain $W$, we have $U_2\neq\emptyset$, and
\[
U_2=\{\ww:(1+\alpha)\underline{s}\leq s(\ww)\leq(1-\gamma)\overline{a}\},
\]
where $s:W\to\R$ is the spread function. Let the boundary of $U_2$ be
\[
\partial U_2=\{\ww:s(\ww)=(1+\alpha)\underline{s}\}\cup\{\ww:s(\ww)=(1-\gamma)\overline{a}\}.
\]

Recall that $\pi:W\to[0,1]^4$ is a continuous function, so $\pi_L:W\to[0,1]$ is also continuous. $\pi_L(\ww)=0$ for all $\ww\in W_M$, and $\pi_L(\ww)=1$ for all $\ww\in W_L$. $\pi_L(\ww)$ is continuous on $U_2$ as a monotonic function of $\log s(\ww)$. $\pi_L(\ww)$ is equal to $0$ if $s(\ww)=(1+\alpha)\underline{s}$, and equal to $1$ if $s(\ww)=(1-\gamma)\overline{a}$. Since $\pi_L(\ww)$ is monotonic \wrt\ $\log s(\ww)$ on $U_2$, it is then strictly increasing \wrt\ $\log s(\ww)$ on $U_2$.

\begin{pro}\label{pro5}
The dynamical trading system $(W,F,\sigma)$ is stochastically stable, and its trajectory of bid-ask pairs will stay almost surely within the region
\[
\{\ww:\underline{s}< s(\ww)< (1+\alpha)^3\underline{s}\}\cap V_2,
\]
if (i) $U_2$ and $V_2$ are nonempty, (ii) $0\leq r_s(U_2)<\alpha(1+\alpha)\underline{s}$ and $r_m(V_2)>\alpha(1+\alpha)(2+\alpha)\underline{s}$, and (iii) $\pi_L(\ww)$ is strictly monotonic \wrt\ $\log s(\ww)$ on $U_2$.
\end{pro}

\begin{proof}
Since $U_2\neq\emptyset$ and its $s$-range is less than $\alpha(1+\alpha)\underline{s}$, we obtain
\[
(1-\gamma)\overline{a}\geq(1+\alpha)\underline{s},\ \text{and}\ (1-\gamma)\overline{a}<(1+\alpha)^2\underline{s}.
\]
Note that
\[
\min_{\ww\in U_2}s(\ww)=(1+\alpha)\underline{s},\ \text{and}\ \max_{\ww\in U_2}s(\ww)=(1-\gamma)\overline{a}.
\]
By the condition (iii), $\pi_L:W\to[0,1]$ satisfies
\[
\pi_L(\ww)=0\ \text{if}\ s(\ww)=(1+\alpha)\underline{s},\ \text{and}\ \pi_L(\ww)=1\ \text{if}\ s(\ww)=(1-\gamma)\overline{a}.
\]
Since $\pi_L+\pi_M=1$, we also have
\[
\pi_M(\ww)=1\ \text{if}\ s(\ww)=(1+\alpha)\underline{s},\ \text{and}\ \pi_M(\ww)=0\ \text{if}\ s(\ww)=(1-\gamma)\overline{a}.
\]

If the initial state $\ww\in\partial U_2$ and $s(\ww)=(1-\gamma)\overline{a}$, then $\pi_L(\ww)=1$, and the limit-type trader will generate a new spread $s(\ww)/(1+\alpha)=(1-\gamma)\overline{a}/(1+\alpha)$, such that
\[
\underline{s}\leq\frac{(1-\gamma)\overline{a}}{(1+\alpha)}<(1+\alpha)\underline{s}=\min_{\vv\in U_2}s(\vv).
\]
If the initial state $\ww\in\partial U_2$ and $s(\ww)=(1+\alpha)\underline{s}$, then $\pi_M(\ww)=1$, and the market-type trader will generate a new spread $(1+\alpha)s(\ww)=(1+\alpha)^2\underline{s}$, such that
\[
(1+\alpha)^2\underline{s}>(1-\gamma)\overline{a}=\max_{\vv\in U_2}s(\vv).
\]
Thus, if the initial state $\ww\in\partial U_2$, it will move into a state in $U_1$ after one period. 

Notice that $\pi_L(\ww)=1$ for all $\ww\in W_L$, and $\pi_M(\ww)=1$ for all $\ww\in W_M$, so for any initial bid-ask pair $\ww\in U_1=W_L\cup W_M$, it will surely move into the region $U_2$ after finite periods, simply as the lattice $\Lambda(\ww)\cap U_1$ has finite nodes. So we only need to consider the case that the initial state $\ww\in U_2\setminus\partial U_2$.

Consider any initial state $\ww\in U_2\setminus\partial U_2$, we have $(1+\alpha)\underline{s}< s(\ww)<(1-\gamma)\overline{a}$. Since $\pi_L(\ww)>0$, $\pi_M(\ww)>0$, $\pi_i(\ww)\neq 0$ for all $i\in\{1,2,3,4\}$. Suppose $\ww$ is updated to $\vv=\ww S_1'$ by a marginal trader of type $\sigma(1)$, $s(\vv)=s(\ww)/(1+\alpha)$, which is greater than $\underline{s}$ and less than $(1-\gamma)\overline{a}/(1+\alpha)$. Since $(1-\gamma)\overline{a}/(1+\alpha)<\min_{\ww\in U_2}s(\ww)$, $\vv\in W_M$. The marginal trader in the next period will be market-type, namely, either type $\sigma(4)$ or type $\sigma(2)$, as $\pi_M(\vv)=1$. If she is of type $\sigma(4)$, $\vv$ will then become $\ww S_1'S_4'=\ww$, since $\{\sigma(1),\sigma(4)\}$ is a minimal periodic block, and $S_4S_1=I$. If she is of type $\sigma(2)$, $\vv$ will become $\vv'=\ww S_1'S_2'=\ww(S_2S_1)'$, where
\[
S_2S_1=
\begin{pmatrix}
1 & 0\\
-\alpha & 1+\alpha
\end{pmatrix}
\begin{pmatrix}
\frac{1}{1+\alpha} & \frac{\alpha}{1+\alpha}\\
0 & 1
\end{pmatrix}
=\frac{1}{1+\alpha}
\begin{pmatrix}
1 & \alpha\\
-\alpha & 1+2\alpha
\end{pmatrix},
\]
so $(-1,1)S_2S_1=(-1,1)$, which implies that $s(\vv')=s(\ww)$. Thus after two periods, $\ww$ will return back to itself or move to a state with the same spread as itself.

If $\ww$ is updated to $\vv=\ww S_2'$ by a marginal trader of type $\sigma(2)$, then $\vv\in W_L$. So in the next period, there will come either type-$\sigma(3)$ trader or type-$\sigma(1)$ trader. The type-$\sigma(3)$ trader will update $\vv$ to $\vv S_3'=\ww S_2'S_3'=\ww$, as $\{\sigma(2),\sigma(3)\}$ is a minimal periodic block, and $S_3S_2=I$. The type-$\sigma(1)$ trader will update $\vv$ to $\ww S_2'S_1'$ that has the same spread as $\ww$. If $\ww$ is updated to $\ww S_3'$ by a marginal trader of type $\sigma(3)$, $\ww S_3'$ will then be updated either to $\ww S_3'S_2'=\ww$, or to $\ww S_3'S_4'$ having the same spread as $\ww$. At last, if $\ww$ is updated to $\ww S_4'$ by a marginal trader of type $\sigma(4)$, $\ww S_4'$ will then move either to $\ww S_4'S_1'=\ww$, or to $\ww S_4'S_3'$ such that $s(\ww S_4'S_3')=s(\ww)$.

Therefore, any initial state $\ww\in U_2\setminus\partial U_2$ will be updated to a state with the same spread as itself after two consecutive periods. Note that the dynamics in $(W,F,\sigma)$ is repeatedly composed of those two-period dynamical blocks, thus the spread that can be achieved in such a dynamical trading system will be greater than
\[
\inf_{\ww\in U_2\setminus\partial U_2}s(\ww)/(1+\alpha)\geq\min_{\ww\in U_2}s(\ww)/(1+\alpha)=\underline{s},
\]
and less than
\[
\sup_{\ww\in U_2\setminus\partial U_2}(1+\alpha)s(\ww)\leq (1+\alpha)\max_{\ww\in U_2}s(\ww)<(1+\alpha)^3\underline{s},
\]
where $\max_{\ww\in U_2}s(\ww)<(1+\alpha)^2\underline{s}$. 

So the trajectory of bid-ask pairs starting from any initial state $\ww\in W$ will be bounded in the region $\{\ww:\underline{s}< s(\ww)<(1+\alpha)^3\underline{s}\}\cap V_2$ almost surely.
\end{proof}

\begin{rmk}
If $0\leq r_s(U_2)<\alpha(1+\alpha)\underline{s}$, there exists a unique exponent $l\in[1,2)$, such that
\begin{equation}
\underline{s}/\overline{a}=\frac{1-\gamma}{(1+\alpha)^l}.
\end{equation}
Note that the upper bound of a bid-ask spread in $W$ is $\overline{a}$, so $(1+\alpha)^3\underline{s}$ should be less than or equal to $\overline{a}$, or equivalently
\[
\underline{s}/\overline{a}\leq\frac{1}{(1+\alpha)^3}.
\]
Thus we need $(1-\gamma)(1+\alpha)^h\leq 1$, where $h=3-l$, so $h\in(1,2]$. It is a slightly stricter requirement than $(1-\gamma)(1+\alpha)\leq 1$ we have used before, since $(1+\alpha)^h>1+\alpha$.

As $U_2=\{\ww:(1+\alpha)\underline{s}\leq s(\ww)\leq (1-\gamma)\overline{a}\}$, we have two nonempty regions in $U_1$, which contain buffering overflows, namely,
\[
\{\ww:\underline{s}<s(\ww)<(1+\alpha)\underline{s}\},\ \text{and}\ \{\ww:(1-\gamma)\overline{a}<s(\ww)<(1+\alpha)^3\underline{s}\}.
\]
Evidently, the trajectory of bid-ask pairs will not stay exactly within the kernel region $K=U_2\cap V_2$, but within $K$ and parts of the buffering region $U_1\cap V_2$ in $H$.
\end{rmk}

As we can see from the proof of Proposition \ref{pro5}, there are eight possible two-period dynamical blocks for all $\ww\in U_2\setminus\partial U_2$. Half of them are minimal periodic blocks, so the bid-ask pair in each of them will return back to $\ww$. The remaining ones will update $\ww$ to $\ww'\neq\ww$, such that $s(\ww')=s(\ww)$.

If $\ww'\neq\ww$, the distance between $\ww$ and $\ww'$ is  
\[
\frac{\sqrt{2}\alpha}{(1+\alpha)}s(\ww),\ \text{or}\ \sqrt{2}\alpha s(\ww),
\]
while the absolute difference between $m(\ww)$ and $m(\ww')$ is
\[
\frac{\alpha}{1+\alpha}s(\ww),\ \text{or}\ \alpha s(\ww). 
\]

The possible trajectories of the dynamical trading system $(W,F,\sigma)$ are consecutive combinations of these two-period dynamical blocks. Recall that $s(\ww)\in((1+\alpha)\underline{s},(1-\gamma)\overline{a})$ for all $\ww\in U_2\setminus\partial U_2$. The trajectory starting from $\ww$ will be bounded within the region $\{\vv:s(\ww)/(1+\alpha)\leq s(\vv)\leq (1+\alpha)s(\ww)\}$, where
\[
(1+\alpha)\underline{s}\in (s(\ww)/(1+\alpha),s(\ww)),\ \text{and}\ (1-\gamma)\overline{a}\in(s(\ww),(1+\alpha)s(\ww)).
\]
So the region containing buffering overflows are
\[
\{\vv:s(\ww)/(1+\alpha)\leq s(\vv)<(1+\alpha)\underline{s}\}\cup\{\vv:(1-\gamma)\overline{a}<s(\vv)\leq(1+\alpha)s(\ww)\},
\]
where $s(\ww)/(1+\alpha)>\underline{s}$, and $(1+\alpha)s(\ww)<(1+\alpha)^3\underline{s}$. We show two trajectories starting from $\ww$ with seven periods in the $b$-$a$ plane.

Let the initial time be $t\in\Z$. We have a sequence with infinite bid-ask vectors, say,
\[
\{\dd_t,\D_{t+1},\D_{t+2},\dotsc,\D_\infty\}.
\]
Let $t=0$, and the spread of $\dd_0$ be $s$. Since all the states at the time $t\in 2\Z_+$ have the same spread $s$, the sequence of bid-ask vectors at the time $t\in 2\Z_+$ has a simple representation. Define $\D^\ast_t=\D_{2t}$ for all $t\in\Z_+$. We obtain a stochastic process $\{\D^\ast_t,\ t\in\Z_+\}$, where
\begin{equation}
\D^\ast_{t+1}=\D^\ast_t+\varepsilon_t(1,1)',
\end{equation}
where $\D_0^\ast=\dd_0$, and $\varepsilon_t\in\R$ for all $t\in \Z_+$. For each $\D^\ast_t=(B_{2t},A_{2t})$, we have $A_{2t}-B_{2t}=s$. So the process $\{\D^\ast_t,\ t\in\Z_+\}$ will stay on the line with a constant spread $s$ in the $b$-$a$ plane.

At each time $t\in\Z_+$, 
\[
\varepsilon_t\in\Big\{\pm\frac{\alpha}{1+\alpha}s,\pm\alpha s\Big\}.
\]
The probability distribution of $\varepsilon_t$ is determined by the function $\pi_B$ \wrt\ the mid-price of $\dd_{2t}$ at the time $2t$ and the function $\pi$ at the time $2t+1$. If the mid-price of $\dd_{2t}$ is close to $s/2$, the probability that $\varepsilon_t<0$ is $0$. If it is close to $\overline{a}-s/2$, then the probability that $\varepsilon_t>0$ is $0$. Thus, $\{\D^\ast_t,\ t\in\Z_+\}$ is bounded within a line segment with $(0,s)$ and $(\overline{a}-s,\overline{a})$ as its end points in the $b$-$a$ plane. Suppose the probability distribution of $\varepsilon_{t}$ at some time $t\in\Z_+$ is uniform, then its mean is $0$, and its volatility is greater than $2\alpha s/(1+\alpha)$ and less than $2\alpha s$.

We can also consider the stochastic bid-ask spread $S_t$ with an initial spread $S_0=s$. If $0\leq r_s(U_2)<\alpha(1+\alpha)\underline{s}$, we have $S_t\in\{s,(1+\alpha)s,s/(1+\alpha)\}$, and $S_{t+1}\neq S_t$ for all $t\in\Z_+$. So the process $\{S_t,\ t\in\Z_+\}$ can be represented by
\begin{equation}
S_{t}=\eta_t s,
\end{equation}
where $\eta_t\in\{1+\alpha,1/(1+\alpha)\}$ for all $t\in\Z_+$. The probability distribution of $\eta_t$ is determined by the function $\pi_L$ \wrt\ the spread of $\dd_t$ at the time $t$. If we come across a uniform distribution at a time $t$, say, $\eta_t$ takes either value with an equal probability, then the mean of $\log\eta_t$ is $0$.

Notice that $\eta_t$ and $\varepsilon_t$ are not independent, since $\varepsilon_t$ is fully captured by $\eta_{2t}$ and $\eta_{2t+1}$ for all time $t\in2\Z_+$. If $\vert\varepsilon_t\vert=\alpha s$, then $\eta_{2t}=1+\alpha$ and $\eta_{2t+1}=1/(1+\alpha)$. If $\vert\varepsilon_t\vert=\alpha s/(1+\alpha)$, then $\eta_{2t}=1/(1+\alpha)$ and $\eta_{2t+1}=1+\alpha$. Thus $\eta_{2t}\eta_{2t+1}=1$ for all $t\in\Z_+$.

\subsection{Controlled Mid-Price Dynamics}

If the $m$-range of $V_2$, rather than the $s$-range of $U_2$, is sufficiently small, we may establish a similar result to Proposition \ref{pro5}. Assume $0\leq r_m(V_2)<\alpha(1+\alpha)\underline{s}/2$. By the condition $r_m(V_2)\geq 0$, we have $V_2\neq\emptyset$ and $W_B\cap W_S=\emptyset$, so $(1+\delta)\underline{s}\leq(1-\epsilon)(2\overline{a}-\underline{s})$, or
\[
\frac{\underline{s}/2}{\overline{a}-\underline{s}/2}\leq\frac{1-\epsilon}{1+\delta}.
\]
On the other hand, if $r_m(V_2)<\alpha(1+\alpha)\underline{s}/2$, then
\[
(1-\epsilon)(2\overline{a}-\underline{s})-(1+\delta)\underline{s}<\alpha(1+\alpha)\underline{s},
\]
which implies
\[
\frac{\underline{s}/2}{\overline{a}-\underline{s}/2}>\frac{1-\epsilon}{1+\delta+\alpha(1+\alpha)}.
\]

As a result, if $0\leq r_m(V_2)<\alpha(1+\alpha)\underline{s}/2$, the domain $W$ will have the following property,
\begin{equation}
\frac{1-\epsilon}{1+\delta+\alpha(1+\alpha)}<\frac{\underline{s}/2}{\overline{a}-\underline{s}/2}\leq\frac{1-\epsilon}{1+\delta}.
\end{equation}
Under that condition, we have a nonempty $V_2$,
\[
V_2=\{\ww:(1+\delta)\underline{s}/2\leq m(\ww)\leq(1-\epsilon)(\overline{a}-\underline{s}/2)\},
\]
where $m:W\to\R$ is the mid-price function. Let the boundary of $V_2$ be
\[
\partial V_2=\{\ww:m(\ww)=(1+\delta)\underline{s}/2\}\cup\{\ww:m(\ww)=(1-\epsilon)(\overline{a}-\underline{s}/2)\}.
\]

Consider again a continuous function $\pi$, so $\pi_B:W\to[0,1]$ is also continuous. $\pi_B(\ww)=1$ for all $\ww\in W_B$, and $\pi_B(\ww)=0$ for all $\ww\in W_S$. Recall that $\pi_B(\ww)$ is monotonic \wrt\ $\log m(\ww)$ on $V_2$. $\pi_B(\ww)=1$ if $m(\ww)=(1+\delta)\underline{s}/2$, and $\pi_B(\ww)=0$ if $m(\ww)=(1-\epsilon)(\overline{a}-\underline{s}/2)$, so $\pi_B(\ww)$ is decreasing \wrt\ $\log m(\ww)$ on $V_2$.

The lower bound of $\underline{s}/(2\overline{a}-\underline{s})$ is a sufficient condition for the existence of dynamical two-period switching blocks, by which any state in $V_2$ at the time $t\in 2\Z_+$ can return back into $V_2$ after two periods. This statement is true, if we can confirm that any $\ww\in\partial V_2$ will be updated to some $\ww'\in V_1=W\setminus V_2$ in the time $t+1$. Note that
\[
\max_{\vv\in V_2}m(\vv)=(1-\epsilon)(\overline{a}-\underline{s}/2),\ \text{and}\ \min_{\vv\in V_2}m(\vv)=(1+\delta)\underline{s}/2,
\]
and they can only be achieved by the states in $\partial V_2$. 

If $\ww\in\partial V_2$ such that $m(\ww)=\max_{\vv\in V_2}m(\vv)$, then $\pi_S(\ww)=1-\pi_B(\ww)=1$. So the next mid-price  generated by a sell-type trader is
\[
m(\ww)-\frac{\alpha}{2}s(\ww),\ \text{or}\ m(\ww)-\frac{\alpha}{2(1+\alpha)}s(\ww).
\]
The largest mid-price should be less than $\min_{\vv\in V_2}m(\vv)$, so that the new state will sufficiently be in $V_1$. Thus we have
\[
\max_{\vv\in V_2}m(\vv)-\frac{\alpha}{2(1+\alpha)}\min_{\ww\in\partial V_2}s(\ww)<\min_{\vv\in V_2}m(\vv).
\]
Note that $s(\ww)/(1+\alpha)\geq(1+\alpha)\underline{s}$ for all $\ww\in U_2\cap V_2$, so $\min_{\ww\in\partial V_2}s(\ww)=(1+\alpha)^2\underline{s}$. Then we obtain
\[
(1-\epsilon)(2\overline{a}-\underline{s})-\alpha(1+\alpha)\underline{s}<(1+\delta)\underline{s},
\]
which is equivalent with the condition $r_m(V_2)<\alpha(1+\alpha)\underline{s}/2$.

If $\ww\in\partial V_2$ such that $m(\ww)=\min_{\vv\in V_2}m(\vv)$, then $\pi_B(\ww)=1$. So the new mid-price generated by a buy-type trader is
\[
m(\ww)+\frac{\alpha}{2}s(\ww),\ \text{or}\ m(\ww)+\frac{\alpha}{2(1+\alpha)}s(\ww).
\]
We require that the smallest mid-price should be greater than $\max_{\vv\in V_2}m(\vv)$, so
\[
\min_{\vv\in V_2}m(\vv)+\frac{\alpha}{2(1+\alpha)}\min_{\ww\in\partial V_2}s(\ww)>\max_{\vv\in V_2}m(\vv),
\]
which is same as the above inequality.

\begin{pro}\label{pro6}
The dynamical trading system $(W,F,\sigma)$ is stochastically stable, and its trajectory of bid-ask pairs will stay almost surely within the region
\[
\{\ww:\underline{m}<m(\ww)<\overline{m}\}\cap U_2,
\]
where $\underline{m}$ and $\overline{m}$ are constants in the trading system, if (i) $V_2$ and $U_2$ are nonempty, (ii) $0\leq r_m(V_2)<\alpha(1+\alpha)\underline{s}/2$ and $r_s(U_2)>\alpha(1+\alpha)(2+\alpha)\underline{s}$, and (iii) $\pi_B(\ww)$ is strictly monotonic \wrt\ $\log m(\ww)$ on $V_2$.
\end{pro}

The proof of Proposition \ref{pro6} is roughly similar to that of Proposition \ref{pro5}, and thus not provided here. The trajectory of bid-ask pairs will again not be bounded within the kernel region $K=V_2\cap U_2$, but within $K$ and parts of the buffering region $V_1\cap U_2$, so that the trading system can contain certain buffering overflows to support its stability. In the limit, each state on the trajectory will be either in $K=V_2\cap U_2$ or in the buffering region $V_1\cap U_2$ with equal probability.

Since the updating process of the mid-price also depends on the spread of the states on the trajectory, the buffering region used to hold the buffering overflows is very large. Actually, we can notice that the $m$-range of the region $\{\ww:\underline{m}<m(\ww)<\overline{m}\}\cap U_2$ is quite wide, as 
\[
\underline{m}=(1+\delta)\underline{s}/2-\alpha(1-\gamma)\overline{a}/2,\ \text{and}\ \overline{m}=(1-\epsilon)(\overline{a}-\underline{s}/2)+\alpha(1-\gamma)\overline{a}/2,
\]
where $(1-\gamma)\overline{a}=\max_{\ww\in U_2}s(\ww)$. Thus its $m$-range is
\[
\overline{m}-\underline{m}=r_m(V_2)+\alpha(1-\gamma)\overline{a},
\]
which is greater than or equal to $\alpha(1-\gamma)\overline{a}$, and less than $\alpha(1-\gamma)\overline{a}+\alpha(1+\alpha)\underline{s}/2$. So $\overline{m}-\underline{m}=O(\overline{a})$. But the $s$-range of the region $\{\ww:\underline{s}<s(\ww)<(1+\alpha)^3\underline{s}\}\cap U_2$ is only $((1+\alpha)^3-1)\underline{s}=O(\underline{s})$, where $O(\underline{s})\ll O(\overline{a})$.

Proposition \ref{pro5} and \ref{pro6} also give us interesting observations on the volatility of the mid-price and the bid-ask spread on a limit-order market. If the spread $s$ in the kernel region $K=U_2\cap V_2$ is required to be \emph{ex ante} stable, then the random trajectory in the dynamical trading system will be bounded within a certain region with a sufficiently small $s$-range. However, if the mid-price $m$ in $K$ is required to be \emph{ex ante} stable, the random trajectory can not remain in a certain region with a small $m$-range. Therefore, we can say the disorder in the bid-ask spread mainly comes from its intrinsic volatility, while the disorder in the mid-price is mainly stored in the information on the market, rather than its intrinsic characteristics.

\section{Further Discussions}

In this work, we take a dynamical perspective to investigate the microstructure of the limit order market. A limit order market is theoretically considered as a dynamical trading system, in which traders' decision processes interact with their different trading types iteratively. In our analysis, the perfect information required to model strategic behaviors is loosened to the knowledge of the so-called atomic trading mechanism between consecutive periods. So we actually have a simplified assumption, but again obtain a powerful ability to understand the dynamical properties of the order flow and the order book in the limit order market. 

As we know, the measures of the liquidity are usually associated with its four dimensions, say, immediacy, width, depth, and resiliency (see \eg\ Harris \cite{harris03}, \S\,19.2). In this work, however, we implicitly do not consider the role of depth in the market liquidity, as we assume that the best quotes in a limit order market always have the same number of shares as any possible coming order. This assumption allows us to study a trading process with only one marginal trader in each trading period. Henceforth, we provide a few interesting and suggestive results on such an induced dynamical trading system. 

If we, alternatively, assume the orders may have different trading volumes with the depth at the quotes in our trading system, we should take account of the possibility that different traders in the population $N$ can form a trading group. Otherwise, if there is no trading group, and all traders in $N$ individually participate in the market, then we would actually only introduce additional jumps and diffusions into our theoretical system of this work, noting that the submitted orders are sequentially operated according to the principle of price-time priority. Thus in an analytical framework allows the existence of trading groups, we can study a limit order market characterized with more dimensions on its liquidity.

\bigskip

\end{document}